\documentclass{sig-alternate}

\usepackage{comment}
\usepackage{graphicx}
\usepackage[caption=false]{subfig}
\usepackage{color}
\usepackage{epsf}
\usepackage[vlined,ruled]{algorithm2e}
\usepackage{xspace}
\usepackage{xfrac}

\usepackage[T1]{fontenc}
\usepackage[T1,hyphens]{url}
\usepackage{underscore}
\usepackage{amsmath,amsfonts,amssymb}

\input{xoverline.texinc}
\newtheorem{lemma}{Lemma}
\newtheorem{corollary}{Corollary}

\newcommand{\RefFigure}[1]{Fig.~\ref{#1}}
\newcommand{\reffig}[1]{fig.~\ref{#1}}
\newcommand{\RefSection}[1]{Section~\ref{#1}}
\newcommand{\refsec}[1]{\S\,\ref{#1}}
\newcommand{\RefAlgorithm}[1]{Algorithm~\ref{#1}}
\newcommand{\refalg}[1]{alg.~\ref{#1}}
\newcommand{\reflemma}[1]{lemma~\ref{#1}}
\newcommand{\refcorr}[1]{corr.~\ref{#1}}

\DeclareMathOperator{\adjList}{adj}

\newcommand{\reftab}[1]{tab.~\ref{#1}}



\begin{document}
%


 \title{Branch-Avoiding Graph Algorithms}


%
%

\numberofauthors{3} 
\author{
\alignauthor
Oded Green\\ 
       \affaddr{Georgia Institute of Technology}\\
       \affaddr{College of Computing} \\
       \affaddr{Atlanta, Georgia, USA} \\
\alignauthor
Marat Dukhan\\ 
       \affaddr{Georgia Institute of Technology}\\
       \affaddr{College of Computing} \\
       \affaddr{Atlanta, Georgia, USA} \\
\alignauthor
Richard Vuduc\\ 
       \affaddr{Georgia Institute of Technology}\\
       \affaddr{College of Computing} \\
       \affaddr{Atlanta, Georgia, USA} \\
}

\maketitle

\begin{abstract}
This paper quantifies the impact of branches and branch mispredictions on the single-core performance for two classes of graph problems. Specifically, we consider classical algorithms for computing connected components and breadth-first search (BFS). We show that branch mispredictions are costly and can reduce performance by as much as 30\%-50\%. This insight suggests that one should seek graph algorithms and implementations that \emph{avoid branches}.

As a proof-of-concept, we devise such implementations for both the classic top-down algorithm for BFS and the Shiloach-Vishkin algorithm for connected components. We evaluate these implementations on current \textsc{x86} and \textsc{ARM}-based processors to show the efficacy of the approach. Our results suggest how both compiler writers and architects might exploit this insight to improve graph processing systems more broadly and create better systems for such problems.

\end{abstract}



\section{Introduction}
\label{sec:intro}

This paper concerns computations on a graph $G=(V,E)$, where $V$ is a set of vertices and $E=\{(u, v)|u, v \in V\}$ is a set of edges.%
  \footnote{In general, the graph may be directed or undirected, weighted or unweighted. These issues do not matter for this study.}
Traditionally, the key challenges associated with creating high-performance graph implementations are computational demand, irregular memory access, difficulty of load balancing, storage, and optimization criteria that cause the problem to be intractable, among others.
In this work, we consider an additional challenge, which is critical to practical implementation but as of yet it is largely unstudied: \emph{branch prediction}, which is an important factor in single-core performance on essentially all modern multi- and emerging manycore processors.
We show subtle and sometimes unexpected performance phenomena that suggest incorrectly predicted branches can reduce single-core efficiency.
These observations suggest that a simple algorithmic redesign, in which branches are avoided, can improve and offer consistent performance.

We are motivated to study branches by the fact that exploiting instruction-level parallelism is critical for achieving high single-core throughput, which is the building block for all higher levels of parallelization (e.g., shared memory or distributed memory).
The presence of a conditional branch interrupts the flow of instructions; if it is not known whether the branch will be taken, the processor cannot know which instruction to fetch next, creating stalls in the pipeline.
To address this problem, a modern processor core tracks the history of a branch, and uses this state to speculatively fetch the next instruction in what it estimates is the most likely outcome.
If it guesses incorrectly, any speculatively executed instructions must be cancelled, causing slowdowns in time and reductions in energy-efficiency.

This paper analyzes two different graph algorithms with respect to their branching behavior:
connected components, based on the classic Shiloach-Vishkin (SV) algorithm~\cite{ShiloachVishkin}, and the classical form of breadth-first search (BFS)~\cite{Cormen2001}, sometimes referred to as the ``top-down'' algorithm~\cite{beamer2012direction}.
SV is a propagation-based algorithm and BFS is a shortest-path algorithm.
The findings of our paper can in principle be extended to both families of algorithms, including All-Pairs Shortest-Paths (APSP)~\cite{Floyd1962,Warshall1962}, betweenness centrality~\cite{Freeman1977,Brandes2001}, and depth-first search~\cite{HopcroftTarjan}, among numerous others.

We quantify the effect of branch mispredictions for SV and BFS, both analytically and empirically.
In our empirical studies, we write and analyze highly-tuned assembly language versions of these algorithms.
We show that SV, which performs an equal amount of work in every iteration, suffers a performance penalty in its early iterations due in part to an increase in the number of branch mispredictions (i.e., \emph{branch misses}).
In its later iterations, when the branch prediction accuracy increases, the performance increases as well.
This observation motivates a \emph{branch-avoiding} algorithm that reduces the number of branches and branch mispredictions that the algorithm incurs, and overall speedups over the highly tuned branch-based assembly implementation.
The variations in per-iteration performance and number of executed instructions of SV essentially goes away in the branch-avoiding version.

BFS also exhibits branch mispredictions, and we develop a branch-avoiding algorithm for it as well.
However, our specific algorithm significantly increases the number of stores operations by more than an order of magnitude.
Consequently, there is no performance win for BFS in contrast to SV. 
Nevertheless, taken together we believe these two cases, SV and BFS, raise a number of intriguing new questions, both about the role of branch-avoidance in algorithm design, whether compilers can produce our hand-generated transformations, and whether additional architectural support could exploit the branching behavior we observe and mitigate cases of performance loss.



\section{Related Work}
\label{sec:related}

This paper focuses on connected components (CC) and breadth-first search (BFS), in part because they are primitive building blocks of higher-level graph analytics. 
Su\-ch analytics include connected components itself~\cite{ShiloachEven,mccollscc}, as well as computing modularity~\cite{newman2004finding}, detecting communities~\cite{newman2004finding,riedy2012parallel}, partitioning graphs~\cite{karypis1995metis}, computing clustering coefficients~\cite{watts1998collective}, computing centrality metrics (e.g., betweenness centrality~\cite{Freeman1977,Brandes2001,GreenStreaming}, closeness centrality~\cite{sabidussi1966centrality}), as well as computing a wide variety of distance based analytics.
A variety of packages implement these analytics, including STIN\-GER~\cite{stinger-tr,stinger-inserts}, GraphCT~\cite{GraphCT,ediger2013graphct}, Ligra~\cite{shun2013ligra}, Pregel~\cite{malewicz2010pregel}, and the Combinatorial BLAS~\cite{Buluc01112011}.
However, the focus of these packages is on exploiting higher-level shared memory multicore, manycore, distributed memory parallelism~\cite{yoo2005scalable,gregor2005parallel,bulucc2011parallel,checconi2012breaking,beamer2013distributed}, and massively multithreaded systems \cite{bader2006designing,barrett2009implementing}.
Thus, our study of low-level single-core behavior and instruction-level parallelism complements and should apply broadly to this large body of existing work.
 
\paragraph{Branch predictors}
The large body of prior work on branch predictors has focused on their design and implementation in hardware; see Smith's survey of strategies~\cite{smith1981study} among other seminal references~\cite{lee1984branch,yeh1991two,yeh1992alternative,sprangle1997agree,lee1997bi,eden1998yags}.
Little is known publicly about the actual implementation of the branch predictors in modern processors, since these are vendor-specific and proprietary.
As such, there is some ongoing empirical research that tries to demystify these implementations using synthetic benchmarks~\cite{milenkovic2002demystifying,fog2013microarchitecture}.
However, with few exceptions, most of the other work on branch prediction evaluates against a general benchmark suites, such as SPECint2006 and SPECfp2006 benchmarks.
Therefore, they do not provide the additional level of understanding possible with a focus on more specific and application-oriented kernels, as in our study.

\paragraph{Performance engineering of graph computations}
There is some work on low-level performance engineering of graph computations.
Green-Marl is domain specific language, which targets shared-memory platforms\cite{hong2012green}.
It emits backend code that manages shared variables using, for instance, atomic instructions; from published code samples, its implementations are branch-based.
Cong and Makarychev present cache-friendly implementations of graph algorithms \cite{cong2012optimizing}.
They quantify how software prefetching improves spatial locality on both the Power7 and the Sun Niagara2. Both systems support multiple threads per core, which can help in memory latency hiding.

For BFS specifically, there are additional studies.
Chhu\-ga\-ni \emph{et el.} present a shared-memory parallel BFS~\cite{chhugani2012fast}.
They focus on reducing cross-socket communication. Their implementation is lock-free.
Merrill and Garland have developed a highly-tuned GPU implementation~\cite{Merrill}.
Beamer \emph{et al.} have proposed algorithmic changes, which they refer to as being direction- optimizing~\cite{beamer2012direction}.
Though there are many interesting ideas in this body of work, we are not aware of a detailed study of the impact of branches.

\paragraph{Graph property characterizations}
Many researchers have characterized high-level properties of real-world graphs, such as the common existence of power-law degree distributions and the small-world phenomenon~\cite{milgram1967small,watts1998collective,BarabasiDiameter,Leskovec2007,BarabasiPower,faloutsos1999power,Broder2000309}.
Our analysis below is justified in part by some of these findings, such as the existence of a large connected component~\cite{Broder2000309}, which has implications for how our target graph computations will behave.

At a lower-level, Burtscher et al. develop metrics to quantify irregularity, with respect to both memory accesses and control-flow~\cite{burtscher2012quantitative}.
They use these metrics to compare different computations, including graph computations, confirming some aspects of conventional wisdom about what we consider ``regular'' versus ``irregular.''
However, it is not clear (to us) how to translate these metrics into actionable transformations of code that improve performance.
%


\section{Branch Prediction}
\label{sec:branchPred}
\newcommand{\StronglyNotTaken}{\textsc{Strongly-Not-Taken}\xspace}
\newcommand{\WeaklyNotTaken}{\textsc{Weakly-Not-Taken}\xspace}
\newcommand{\Taken}{\textsc{Taken}\xspace}
\newcommand{\NotTaken}{\textsc{Not-Taken}\xspace}
\newcommand{\StronglyTaken}{\textsc{Strongly-Taken}\xspace}
\newcommand{\WeaklyTaken}{\textsc{Weakly-Taken}\xspace}

 \begin{figure}
  \includegraphics[width=\columnwidth,clip=true,trim=.5in 0in .5in 0in]{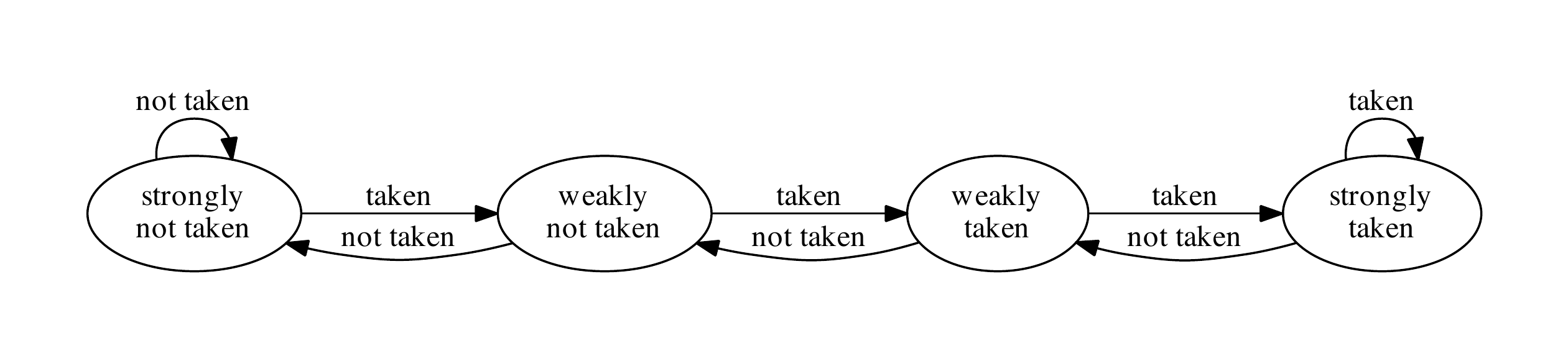}
  \caption{
    A 2-bit branch predictor behaves as shown in this finite-state automaton.
    Each node is a state representing the next prediction, e.g., the strongly and weakly taken states predict ``taken,'' the others, ``not taken.''
    Each edge shows how the state changes once the actual branch condition is resolved.
  }
  \label{fig:branch-pred-fsa}
\end{figure} 

Given a particular (static) \emph{conditional}%
  \footnote{As opposed to an unconditional branch, which always jumps and therefore does not need to be predicted.}
branch in a graph algorithm, our analysis goal is to estimate how many times the branch predictor will mispredict it.%
Although we do not know exactly what type of branch predictor a vendor implements, our analysis assumes a \emph{2-bit branch predictor}~\cite{smith1981study}.
The empirical evaluation of \refsec{sec:experiment} will justify this choice.
Like most branch prediction techniques, it uses the history of previous executions of a given branch to predict the next outcome;%
  \footnote{For instance, a simple \emph{1-bit} predictor predicts that if the last occurrence of a given branch was taken, then so will the next one.}
as such, one may formalize the analysis of predictors mathematically using Markov chains and reason about expected branch misses, which we have done.
However, for concerns of readability and space, this paper omits the details of such analysis, instead stating the key results and offering more intuitive high-level explanations.

\begin{figure*}[t]
\centering
	\subfloat[]{\includegraphics[width=0.2\textwidth]{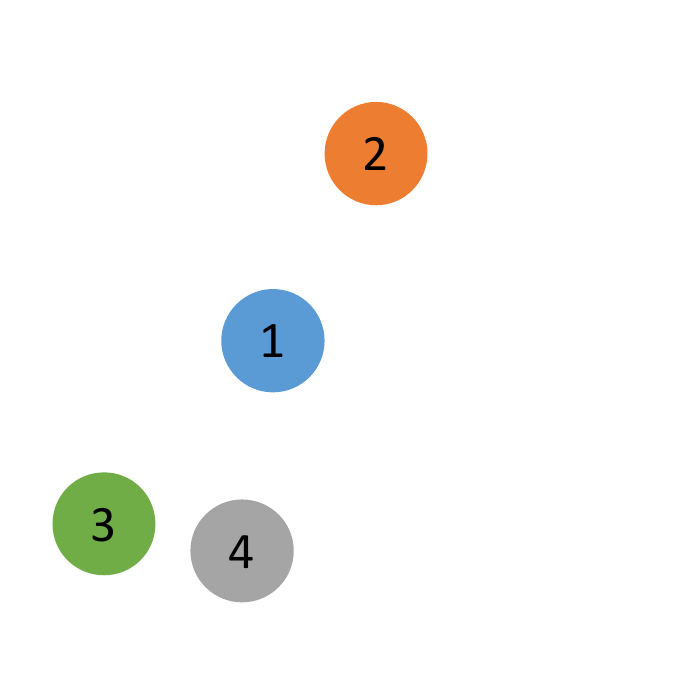}}
	\subfloat[]{\includegraphics[width=0.2\textwidth]{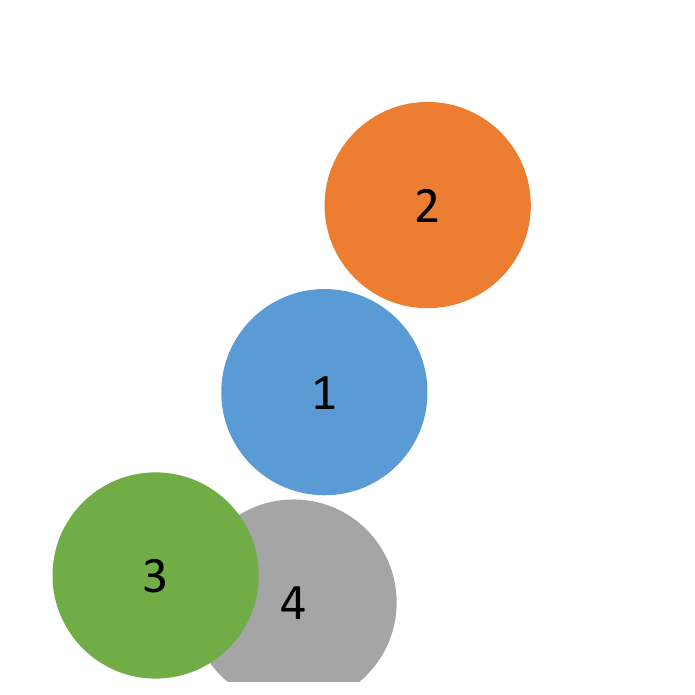}}
	\subfloat[]{\includegraphics[width=0.2\textwidth]{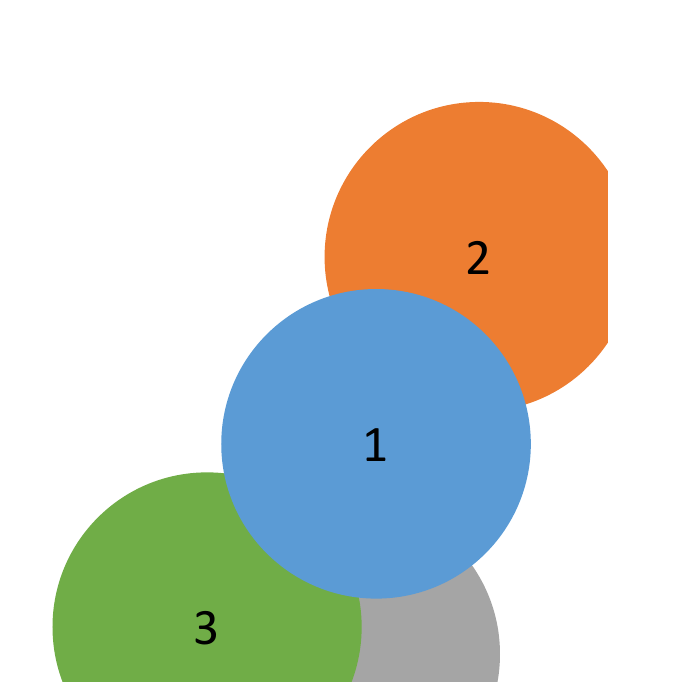}}
	\subfloat[]{\includegraphics[width=0.2\textwidth]{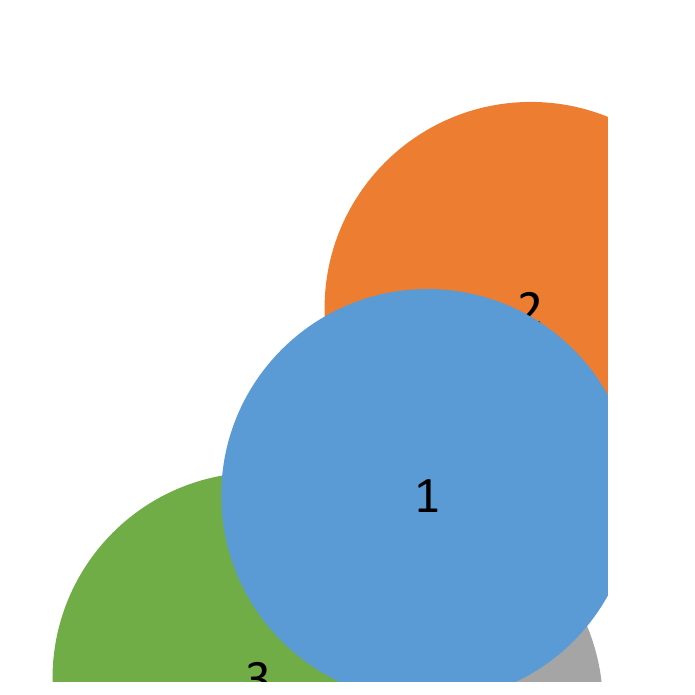}}
	\subfloat[]{\includegraphics[width=0.2\textwidth]{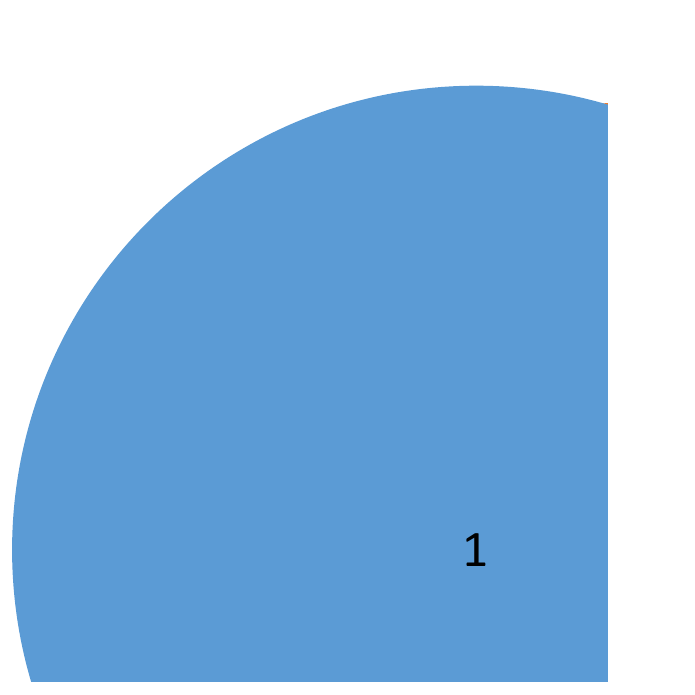}}

  \caption{These sub-figures conceptually show how the connected component id propagates through the graph as time evolves - each subfigure is for a different iteration of the algorithm. This example assumes that all vertices are connected and for simplicity shows only the connected components are 1 through 4. Initially the number of connected components is equal to the number vertices. (a) Depicts the initial state in which each vertex is in its own component. (b)-(d) depict that some vertices belong to the same connected component yet may require multiple label updates (in either the same iteration or a separate iteration).
	(e) is the final state in which there is a single connected component.}
  \label{fig:svPropagation}
\end{figure*}

\begin{algorithm}[t]
  Let $n \geq 0$\;
  $i \leftarrow 0$\;
  \While{$i < n$}{
    \tcp*[h]{\ldots\ $i$, $n$ unmodified; no early exits \ldots}
    \newline
    $i \leftarrow i + 1$\;
  }
  \caption{A simple sequential while-loop, which executes its body exactly $n$ times}
  \label{alg:loop}
\end{algorithm}

\subsection{A model of 2-bit predictors}
For each (static) conditional branch in the program, a 2-bit predictor maintains a 2-bit state value, which encodes four possible states.
Each state value is a prediction for the next occurrence of this branch;
once the true branch condition is known, this state is updated.
The precise states and transitions appear in the finite-state automaton (FSA) of \reffig{fig:branch-pred-fsa}.
In particular, there are four possible states, named \StronglyTaken, \WeaklyTaken, \WeaklyNotTaken, and \StronglyNotTaken.
The ``strong'' states reflect that the last few branches were all the same, i.e., all ``taken'' or all ``not taken,'' and so it is likely the next branch will be the same.
The weak states allow for the predictor's bias to change if a new pattern emerges.

We will further assume that the processor has enough branch state storage to track, for each conditional branch of interest, its 2-bit state for the duration of the program.
That is, we will not consider the case when the processor runs out of branch state storage and must ``evict'' (and therefore losing or resetting) the branch state.

\subsection{Analysis of simple loops}
\label{subsec:bpandga}

Several common programming patterns in graph algorithms are \emph{simple sequential loops}, which iterate over, the set of vertices, edges, neighbors of a vertex, i.e., the \emph{adjacency list}.
Consider, for example, the \emph{simple sequential while-loop} of \refalg{alg:loop}.
By ``simple,'' we mean that (a) the iteration variable $i$ increases monotonically by 1 at each iteration; (b) the loop bound $n$ is constant as the loop executes; and (c) there are no early exits.
Thus, this loop executes its body exactly $n$ times.
The conditional branch in this case evaluates the condition, $i < n$.
We will assume the convention for this loop is such that the branch is taken when the condition is true, and not taken when the condition is false.%
  \footnote{This choice is arbitrary and depends on the specific code generated. There is an equivalent argument if one assumes code such that the branch is taken only when the condition is false.}
There will be exactly $n+1$ evaluations of this branch, only the last of which is \emph{not} taken in order to exit the loop.%
  \footnote{There is an additional branch at the bottom of the loop. However, this branch is \emph{unconditional}, since it \emph{must} jump back to the top of the loop.}
We can state a number of facts about such loops, assuming the 2-bit branch predictor.

\begin{lemma}
When $n \geq 3$, the final state of the 2-bit predictor is \WeaklyTaken.
\label{lemma:afterWeaklyTaken}
\end{lemma}

\begin{proof}
The conditional branch is taken $n$ times.
In the worst case, we begin the loop in the \StronglyNotTaken state.
According to the FSA of \reffig{fig:branch-pred-fsa}, after three taken state transitions, the predictor will be in the \StronglyTaken state.
Since the final branch is \emph{not} taken, the predictor must move into the \WeaklyTaken state.
\end{proof}

\begin{lemma}
When $n \geq 3$, the maximum number of branch mispredictions incurred by the loop's conditional test (ignoring conditional branches in the body) is 3.
\label{lemma:adjTraversal}
\end{lemma}
\begin{proof}
As with \reflemma{lemma:afterWeaklyTaken}, the the initial state of the predictor may be \StronglyNotTaken, which will cause 2 mispredictions before reaching either of the \Taken states.
For the last loop iteration, when $i=n \geq 3$, the predictor will be in the \StronglyTaken state but branch will be taken, incurring one more branch miss.
Thus, there could be up to 3 misses.
Furthermore, there must be at least 1 branch miss, which occurs on the last (not taken) branch;
the reason is that the predictor must be in the \StronglyTaken state by iteration $i=n-1$, independent of the initial state.
\end{proof}

\begin{lemma}
Suppose we execute the same loop $k \geq 2$ times (such as in the case of nested loops), where $n \geq 3$ on the first execution, and $n \geq 1$ on every subsequent execution. These are nested loops - $k$ designates the outer-loop and $n$ for the inner-loop.
Then there may be up to $k+2$ mispredictions for the inner loop, that is, up to 3 misses during the first execution and 1 additional miss on each of the $k-1$ remaining executions.
\label{lemma:nestedLoop}
\end{lemma}

\begin{proof}
Based on \reflemma{lemma:afterWeaklyTaken} the branch predictor is in the \WeaklyTaken state at the end of the first execution of the loop and may see up to 3 mispredictions.
This state becomes the initial state for the next execution.
If $n \geq 1$ on every execution after the first, then the predictor will move to the \StronglyTaken state; on the last iteration, it will return to the \WeaklyTaken state, incurring 1 misprediction.
That is, we will bounce back-and-forth between \StronglyTaken and \WeaklyTaken.
\end{proof}

\begin{corollary}
If $k \gg 2$, we should expect approximately $k$ branch misses.
\label{corr:repeated-loop-misses}
\end{corollary}

\begin{lemma}
Suppose $n=0$. Then the predictor will move toward the \StronglyNotTaken state and cannot be in the \StronglyTaken state;
furthermore, it will incur either 0 or 1 branch misses.
\end{lemma}

\begin{lemma}
Suppose $n=1$. Then the predictor will return to its initial state, incurring either 1 or 2 branch misses.
\end{lemma}

\begin{lemma}
Suppose $n=2$. Then the branch predictor must end in either the \WeaklyTaken or \WeaklyNotTaken states, and will incur between 1 and 3 branch misses.
\end{lemma}


\section{Connected Components}
\label{sec:algSV}

For the problem of finding connected components, we assume the Shiloach and Vishkin (SV) algorithm~\cite{ShiloachVishkin}.
It has been implemented on numerous multiprocessor systems, including the massively threaded Cray XMT~\cite{GraphCT,ediger2013graphct} and a variety of \textsc{x86} systems~\cite{mccollscc}. 

SV is based on a propagation technique, and its pseudocode appears in \refalg{alg:svBranch}.
It maintains for each vertex $v$ a component label, $CC_{id}[v]$, and updates this label to place adjacent vertices into the same connected component.
Initially, each vertex $v$ is placed into a connected component by itself, which by convention is a label equal to the vertex number.
As such, there are a total of $|V|$ connected components at this stage.
In the first iteration, each vertex $v$ compares its own label with each of its neighbors, $u \in \adjList(v)$.
Again by convention, the vertex replaces its own label with the minimum label among itself and its neighbors.
The algorithm is iterative and stops when no further label changes occur, maintained by a flag.

Each iteration requires $O(|V|+|E|)$ computations, since the algorithm accesses all vertices and their respective adjacencies. The maximal length of propagation is limited by the graph diameter $d$. As such, the total time complexity of the algorithm is $O(d \cdot (|V|+|E|))$.
Relative to \refalg{alg:svBranch}, there is a shortcut that can reduce the number of iterations to $d/2$~\cite{ShiloachVishkin}. However, this does not change the asymptotic time complexity of the algorithm, and we do not consider it further.


Conceptually, the component labels propagate as \reffig{fig:svPropagation} depicts.
Initially (a), four of the components have minimal labels locally;
these labels propagate gradually, and the label of a given node may change several times, (b)-(e), possibly even within the same iteration.
Eventually, the algorithm reaches a final state (e) where for a fully-connected graph there will be a single connected component.

\begin{algorithm}[t]

\scriptsize

    // Algorithm initialization 
    
    \For {$v \in V$}
    {
        $CC_{id}[v] \leftarrow v$
    }
    
    $change \leftarrow 1$
    
    // Connected component labeling
    
    \While{$change \neq 0$}
    {
        $change \leftarrow 0$

        \For {$v \in V$}
        {

            $c_v \leftarrow CC_{id}[v]$

            \For {$u \in Neighbors[v]$}
            {

                $c_u \leftarrow CC_{id}[u]$

                \If{$c_u \leq c_v$}
                {
                    $CC_{id}[v] \leftarrow c_u$

                    $change \leftarrow 1$
                }
            }
        }

    }
\caption{Branch-based Shiloach-Vishkin algorithm for finding connect components.}
\label{alg:svBranch}

\end{algorithm}

\begin{algorithm}[t]

\scriptsize

    // Algorithm initialization 
    
    \For {$v \in V$}
    {
        $CC_{id}[v] \leftarrow v$
    }
    
    $change \leftarrow 1$
    
    // Connected component labeling
    
    \While{$change \neq 0$}
    {
        $change \leftarrow 0$

        \For {$v \in V$}
        {

            $c_v^{init} \leftarrow CC_{id}[v]$

            $c_v \leftarrow c_v^{init}$

            \For {$u \in Neighbors[v]$}
            {

                $c_u \leftarrow CC_{id}[u]$

                \If{$c_u \leq c_v$}
                {
                    $c_v \leftarrow c_u$
                }

            }

            $CC_{id}[v] \leftarrow c_v$
            
            $change \leftarrow change \lor c_v \oplus c_v^{init}$
        }

    }
\caption{Branch-avoiding Shiloach-Vishkin algorithm for finding connect components.}
\label{alg:svBranchAvoiding}

\end{algorithm}

\subsection{Branch (Mis)predictions in SV}

The standard version of the SV algorithm (\refalg{alg:svBranch}) has four static conditional branches.
To analyze the branch mispredictions, we assume the 2-bit branch predictor model of \refsec{sec:branchPred}.

The first conditional branch is the termination test of the \emph{while} statement. This condition is evaluated $d+1$ times, where $d$ is the diameter of the graph.
Per \refsec{sec:branchPred}, assuming $d \geq 3$, it should incur at most 3 mispredictions, ignoring mispredictions in the body of the loop.

Next, consider the two conditional branches associated with the two for-loops.
The first for-loop iterates over all vertices;
the second for-loop iterates over all neighbors of each vertex, thereby effectively visiting all edges.
From the facts of \refsec{subsec:bpandga}, the first for-loop will incur up to 3 branch misses in total, assuming sufficiently large $|V|$.
The second for-loop is an instance of a repeated loop (see \reflemma{lemma:nestedLoop}), which is executed $|V|$ times.
Though the exact behavior of the inner loop depends on the degree distribution, we can estimate the misses by
applying \refcorr{corr:repeated-loop-misses}, which implies approximately $|V|$ branch misses.

Finally, the if-statement is the hardest to analyze offline.
The actual number of branch mispredictions will depend on the input graph.
To get an idea, consider the example in \reffig{fig:svPropagation}. In the first iterations, vertices are likely to ``swap'' their connected components multiple times, which complicates branch prediction as there may not be a regular pattern.
As iterations proceed, labels begin to stabilize, making this condition more predictable.
Thus, we should expect to see many mispredictions initially, gradually decreasing as iterations proceed.

\subsection{Branch-avoiding SV}

Algorithm \ref{alg:svBranchAvoiding} presents the pseudocode for a branch-avoiding algorithm for the Shiloach-Vishkin connected component formulation. This algorithm compares the values of the connected component ids; however does not branch based on the value of the comparison. Instead this approach uses a conditional move that copies the value into the variable $c_v$ if and only if the id of $u$ is smaller than the the value in $c_v$.  For the Shiloach-Vishkin algorithm the value of the connected component of $v$ is stored in $c_v$ which is a register, meaning that the number of writebacks (stores) is $|V|$.
To ensure the correctness of the algorithm and that the algorithm will stop at some point, the variable $change$ is updated using a bitwise OR of bitwise XOR between the initial $c_v^{init}$ and the updated $c_v$. If the value of the connected component changed for the current vertex, $c_v$ is not equal to $c_v^{init}$ and their XOR value is non-zero. Accordingly, if any connected component changed, $change$ variable will have non-zero after traversing vertices in $V$.

\section{Breadth First Search}
\label{sec:algBFS}

\begin{algorithm}[t]
\scriptsize
  $Q \leftarrow$ empty queue\;
  initialize $d[v \in V] \leftarrow \infty$\;
  enqueue $r \rightarrow Q$\;
  set $d[r] \leftarrow 0$\;
  \While {$Q$ not empty} {
    dequeue $v \leftarrow Q$\;
    \For {all neighbor $w$ of $v$} {
	
			\If(// $w$ found for the first time) {$d[w] = \infty$}{
	$enqueue \, w \rightarrow Q$\;
	$d[w] \leftarrow d[v] +1$\;
      }  
    }
  }
\caption{Pseudocode for branch-based Breadth First Search starting from a vertex $r$}
\label{alg:algBFS}

\end{algorithm}

\begin{algorithm}[t]
\scriptsize

  $Q \leftarrow$ empty queue\;
	$Q_{len} \leftarrow 1$;
  initialize $d[v \in V] \leftarrow \infty$\;
  enqueue $r \rightarrow Q$\;
  set $d[r] \leftarrow 0$\;
  \While {$Q$ not empty} {
    dequeue $v \leftarrow Q$\;
		$next\_level \leftarrow d[v]+1$;
    \For {all neighbor $w$ of $v$} {
			
			$LOAD (temp, d[w])$\;
			$CMP (temp,d[v])$\;
			$Q[Q_{len}] \leftarrow w$\;
			$COND\_MOVE\_GREATER (temp, next\_level)$\;
			$COND\_ADD(Q_{len}, 1)$\;
			$STORE(temp,d[w])$\;

    }
  }

\caption{Pseudo code for branch-avoiding Breadth First Search}
\label{alg:algBFSBranchAvoiding}

\end{algorithm}

Given a graph and a root vertex $r$, a breadth-first search (BFS) computes the distance of every node in the graph to $r$.
The pseudocode for the BFS algorithm we consider appears in \refalg{alg:algBFS}.

\subsection{Branch (Mis)predictions in BFS}

The BFS algorithm of \refalg{alg:algBFS} has three branches: \emph(while), \emph(for), and \emph(if). We consider the impact of branch mispredictions assuming a 2-bit branch predictor. In particular, we estimate a practical lower bound on the number of branch misses.
We validate this estimate in \refsec{sec:experiment}.

The first branch, \emph(while), is used to iterate through the queue of vertices in the current and next frontiers. For a breadth first search that finds the vertices $\hat{V}$, s.t. $\hat{V} \subseteq V$ , the condition of the \emph{while} branch will be evaluated a total of $|\hat{V}|+1$ times. Based on the lemmas of \refsec{subsec:bpandga} the number of branch misprediction for this statement is $O(1)$. 

\begin{figure*}[tp]
	\centering
		\includegraphics[width=0.95\textwidth]{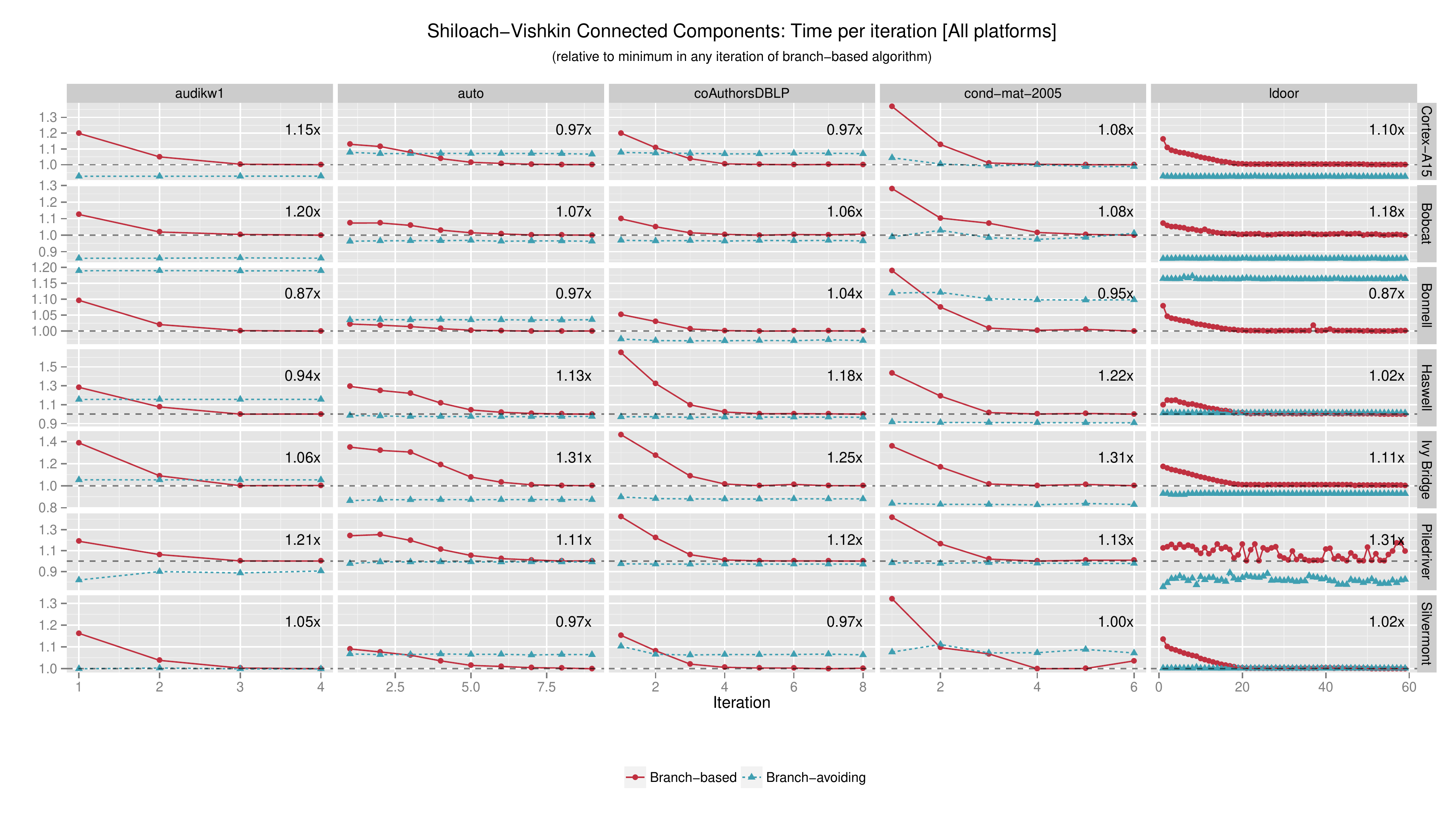}
	\caption{Time as a function of the iteration for the Shiloach-Vishkin algorithm.}
	\label{fig:sv--time--vs--iteration}
\end{figure*}

\begin{figure*}[tp]
	\centering
	\includegraphics[width=0.95\textwidth]{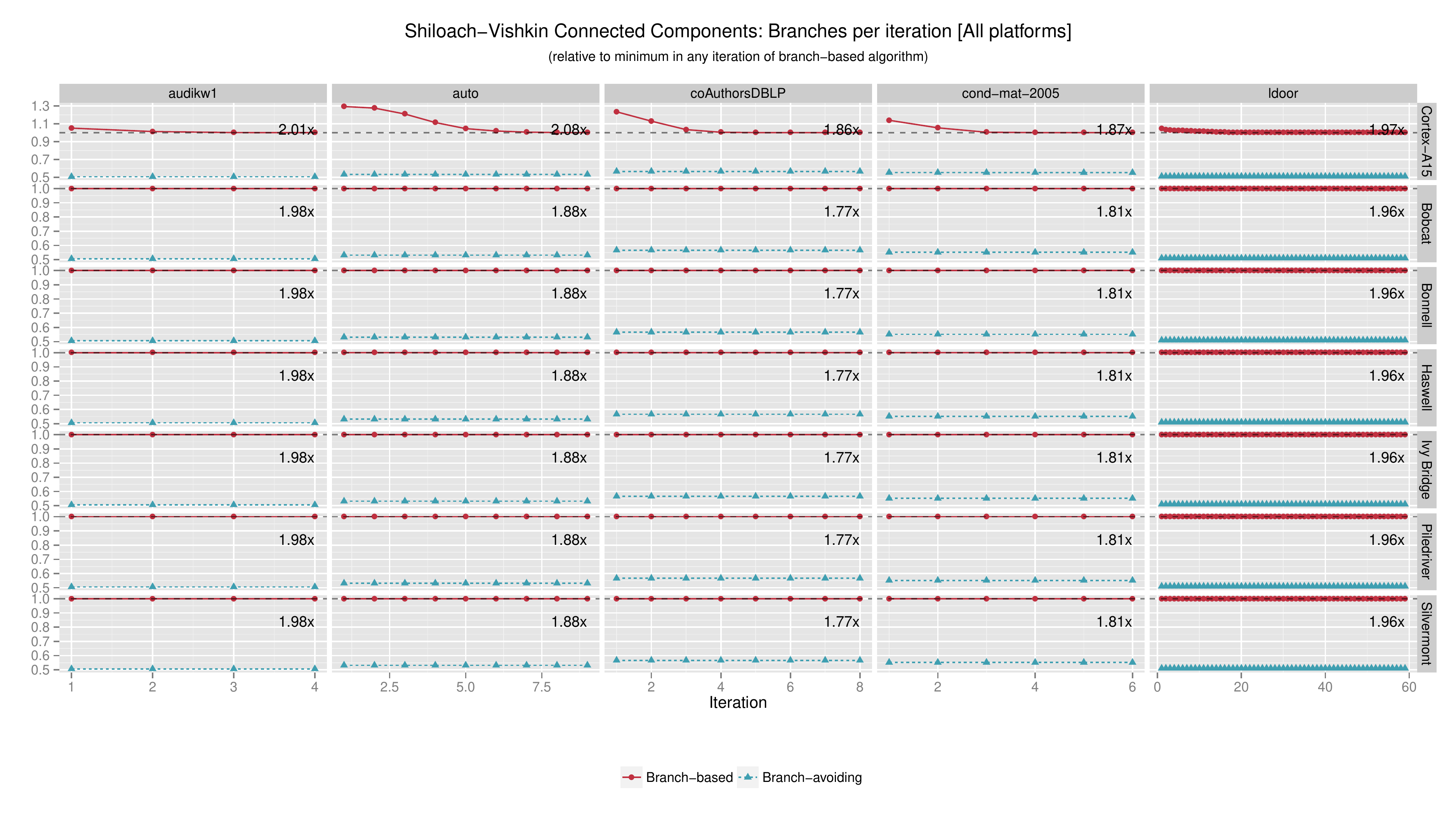}

	\caption{Branches as a function of the iteration for the Shiloach-Vishkin algorithm.}
	\label{fig:sv--branches--vs--iteration}
\end{figure*}

\begin{figure*}[t]
	\centering
		\includegraphics[width=0.95\textwidth]{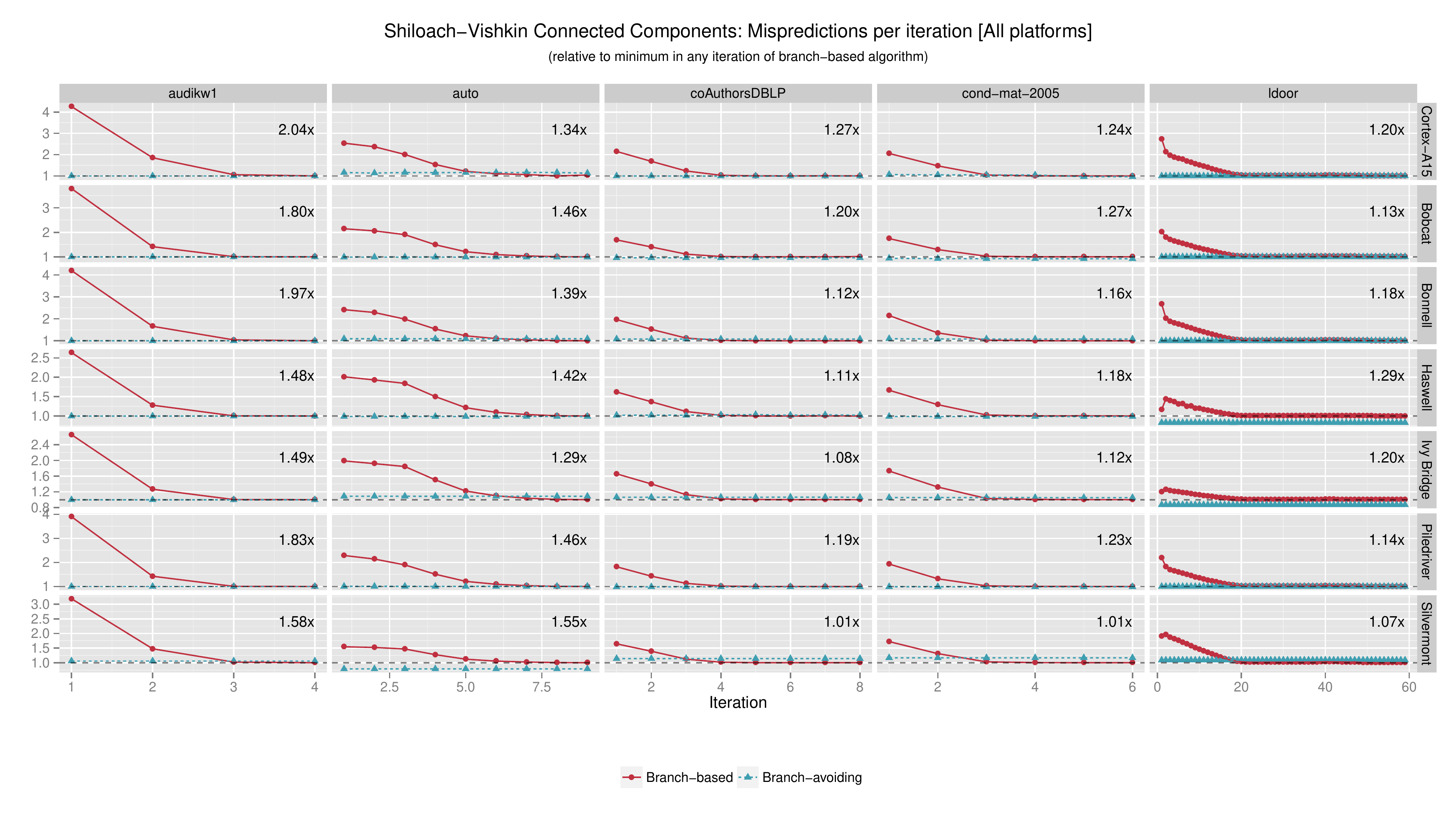}

	\caption{Branch mispredictions as a function of the iteration for the Shiloach-Vishkin algorithm.}
	\label{fig:sv--mispredictions--vs--iteration}
\end{figure*}

The second conditional branch, the \emph{for} statement, is responsible for traversing the adjacency list of a given vertex. Based on the lemmas from \refsec{subsec:bpandga}, as this loop is executed $\hat{V}$ times, there should be approximately $\hat{V}$ misses.
For undirected graphs, each vertex found in the BFS traversal has at least on vertex that will be traversed. For directed graphs a situation can arise that a vertex does not have any outbound adjacencies. 
In practice, there should be nearly $|\hat{V}|$ branch misses for the \emph{for} statement.

The last of these conditions, the \emph{if} statement, checks if a vertex has been found. As each vertex can be found only once, this branch will be taken at most $|\hat{V}|$ times. This branch statement is evaluated $|\hat{E}|$ times where $\hat{E} \subseteq E \wedge \hat{E}={((v,u)|v\in \hat{V}})$. The exact order in which this branch will be $TAKEN$ is highly dependent on the order in which the vertices and edges are accessed.
For example, when the neighbors of the root are traversed, they will all be added to the queue (i.e. branch TAKEN) and there will be fewer misses. 
Now consider the neighbors of the root, these might have some common neighbors in the second frontier, however, only the first traversal of the common neighbor will take the branch. 
In the worst case a scenario can arise in which the state of the \emph{if} branch moves back and forth between WEAKLY-NOT-TAKEN and WEAKLY-TAKEN. This can potentially double the number of branch misses for the \emph{if} branch. As such the \emph{if} statement can have upto $2\cdot|\hat{V}|$ branch misses. 

Thus the upper-bound on the number of branch misses for the branch-based algorithm is $3\cdot|\hat{V}| + O(1)$.

\subsection{Branch Avoiding BFS}

\RefAlgorithm{alg:algBFSBranchAvoiding} presents the pseudocode for a BFS branch-avoiding algorithm. Similar to the Shiloach-Viskin algorithm a compare without branch is used to compare the level of the current vertex with the level of its adjacency. The adjacent vertex, $w$, is placed at the end of the queue. Based on hardware flags, two conditional operations are completed. The first will conditionally move the distance to the vertex if it is found for the first time. The second conditional operation increases the size of the queue. If an element is new then the queue size is increased and the element is enqueued. If the vertex is not new, then it has been placed ``outside'' the queue and will be overwritten when a new vertex is found. When this is done the distance of $w$ is written back to. This means that there are $O(E)$ writebacks.

\section{Empirical Results}
\label{sec:experiment}
\newcommand{\refarndale}{\textsf{Cortex-A15}\ }
\newcommand{\refbobcat}{\textsf{Bobcat}\ }
\newcommand{\refbonnell}{\textsf{Bonnell}\ }
\newcommand{\refhsw}{\textsf{Haswell}\ }
\newcommand{\refivb}{\textsf{Ivy Bridge}\ }
\newcommand{\refpld}{\textsf{Piledriver}\ }
\newcommand{\refslv}{\textsf{Silvermont}\ }

\begin{figure*}[tp]
	\centering
		\includegraphics[width=0.95\textwidth]{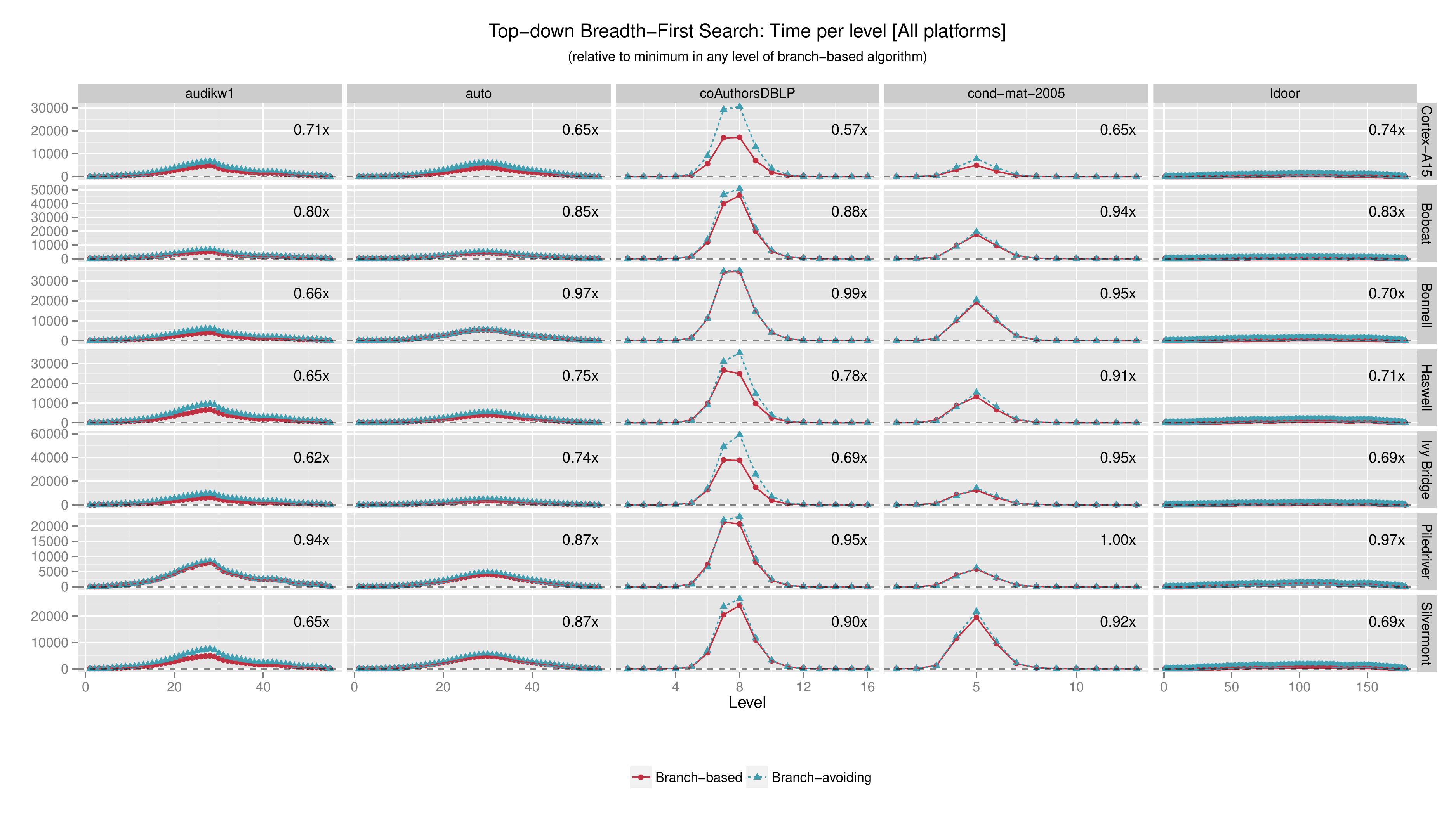}
  	\caption{Time as a function of the iteration for BFS.}
	\label{fig:bfs--time--vs--iteration}
\end{figure*}

\begin{figure*}[tp]
	\centering
	\includegraphics[width=0.95\textwidth]{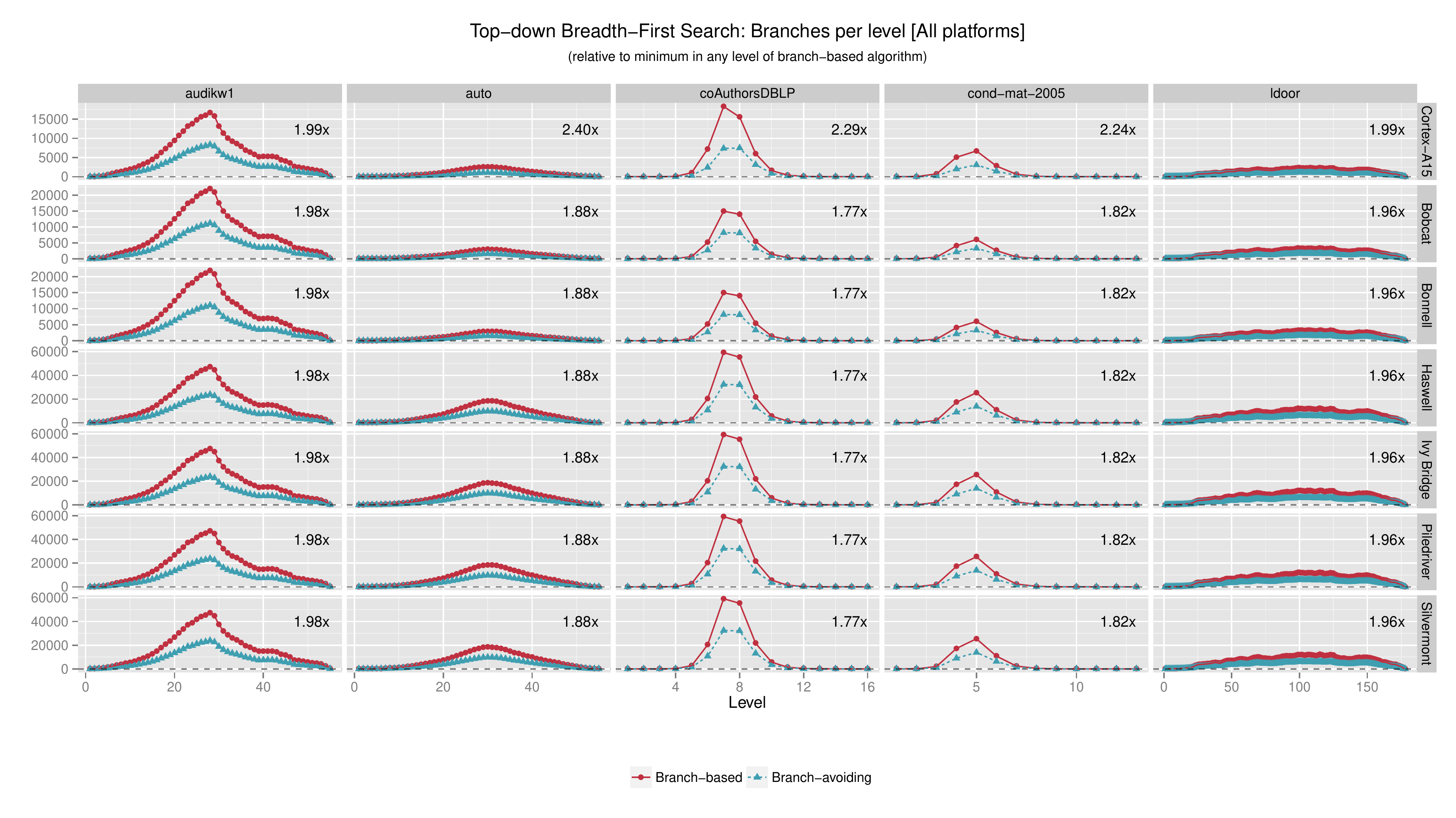}

	\caption{Branches as a function of the iteration for BFS.}
	\label{fig:bfs--branches--vs--iteration}
\end{figure*}

\begin{figure*}[t]
	\centering

	\includegraphics[width=0.95\textwidth]{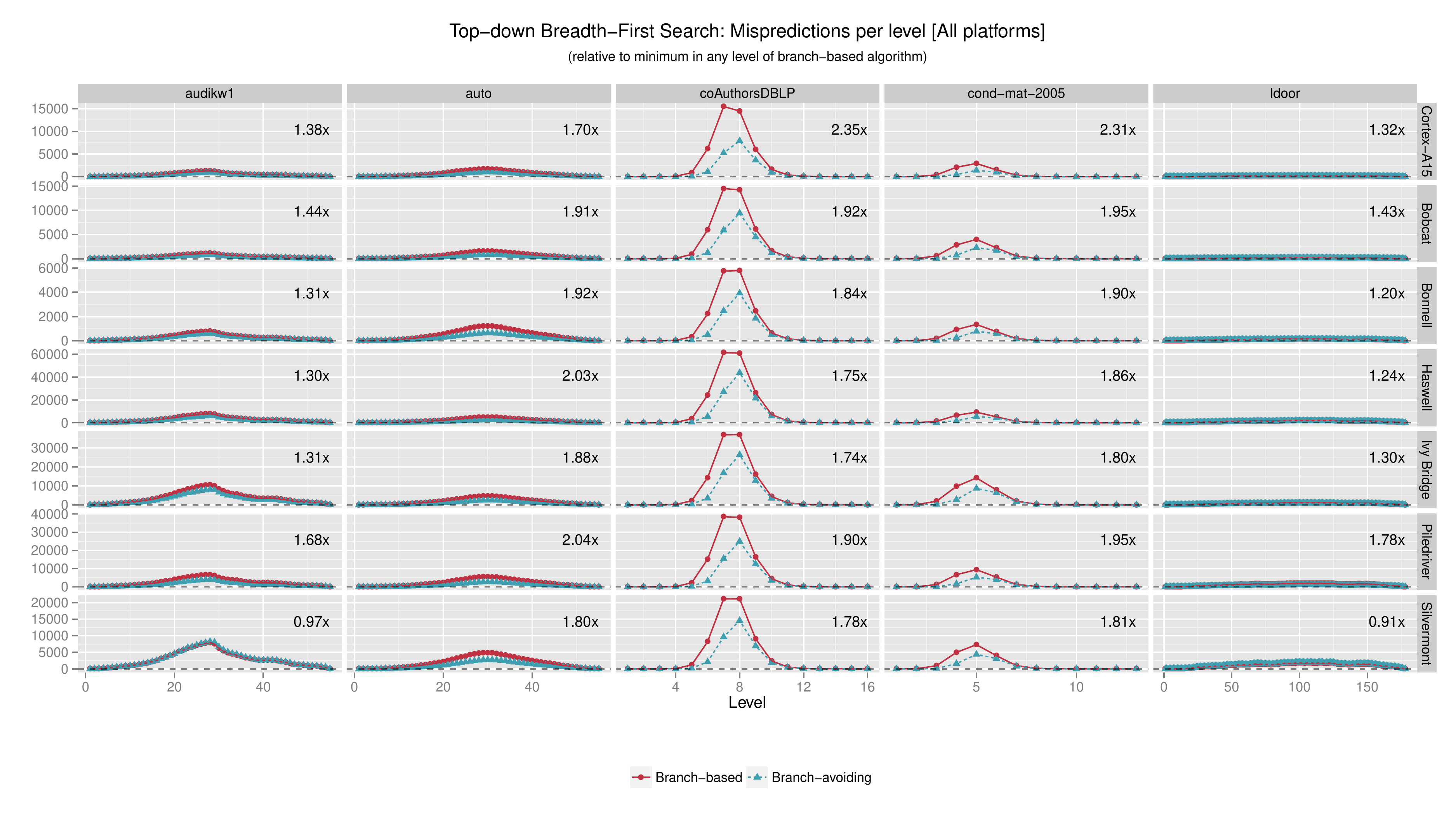}

	\caption{Branch mispredictions as a function of the iteration for BFS.}
	\label{fig:bfs--mispredictions--vs--iteration}
\end{figure*}

\begin{table*}[t!]
\caption{Systems used in experiments.}
\tiny
\centering
\begin{tabular}{|c|c|c|c|c|c|c|c|c|} \hline 
Architecture & Microarchitecture & Processor           & Frequency & L1 Cache & L2 Cache & L3 Cache & DRAM Type    \\ \hline \hline
ARM v7-A     & Cortex-A15 [arn]       & Samsung Exynos 5250 & 1.7 GHz   & 32 KB    & 1 MB     &          & SC DDR3-800  \\ \hline 
x86-64       & Piledriver [pld]       & AMD FX-6300         & 3.5 GHz   & 16 KB    & 2 MB     & 8 MB     & DC DDR3-1600 \\ \hline 
x86-64       & Bobcat            & AMD E2-1800         & 1.7 GHz   & 32 KB    & 512 KB   &          & SC DDR3-1333 \\ \hline 
x86-64       & Haswell [hsw]          & Intel Core i7-4770K & 3.5 GHz   & 32 KB    & 256 KB   & 8 MB     & DC DDR3-2133 \\ \hline 
x86-64      & Ivy-Bridge [ivb]       & Intel Core i3-3217U & 1.8 GHz   & 32 KB    & 256 KB   & 3 MB     & DC DDR3-1600 \\ \hline 
x86-64       & Silvermont [slv]       & Intel Atom C2750    & 2.4 GHz   & 24 KB    & 1 MB     &          & DC DDR3-1600 \\ \hline 
x86-64      & Bonnell           & Intel Atom 330      & 1.6 GHz   & 24 KB    & 512 KB   &          & SC DDR3-800  \\ \hline 
\end{tabular}

\label{tab:systems}
\end{table*}

\begin{table}[t!]
\begin{center}
\small

\caption{Graphs from the 10th DIMACS Implementation Challenge used in our experiments.}
\label{tab:dimacs}
\begin{tabular}{|l|c|c|c|}\hline
Name  & Graph Type &$|V|$ & $|E|$ \\\hline \hline

audikw1 &Matrix& 943,695 & 38,354,076\\\hline
auto &Partitioning& 448,695, & 3,314,611 \\\hline
coAuthorsDBLP  &Collaboration&299,067 & 977,676\\\hline
cond-mat-2005&Clustering & 40,421 & 175,691\\\hline
ldoor & Matrix & 952,203& 22,785,136\\\hline 



\end{tabular}

\end{center}
\end{table}

\begin{figure*}[t]
	\centering
		\subfloat[]{\includegraphics[width=0.48\textwidth]{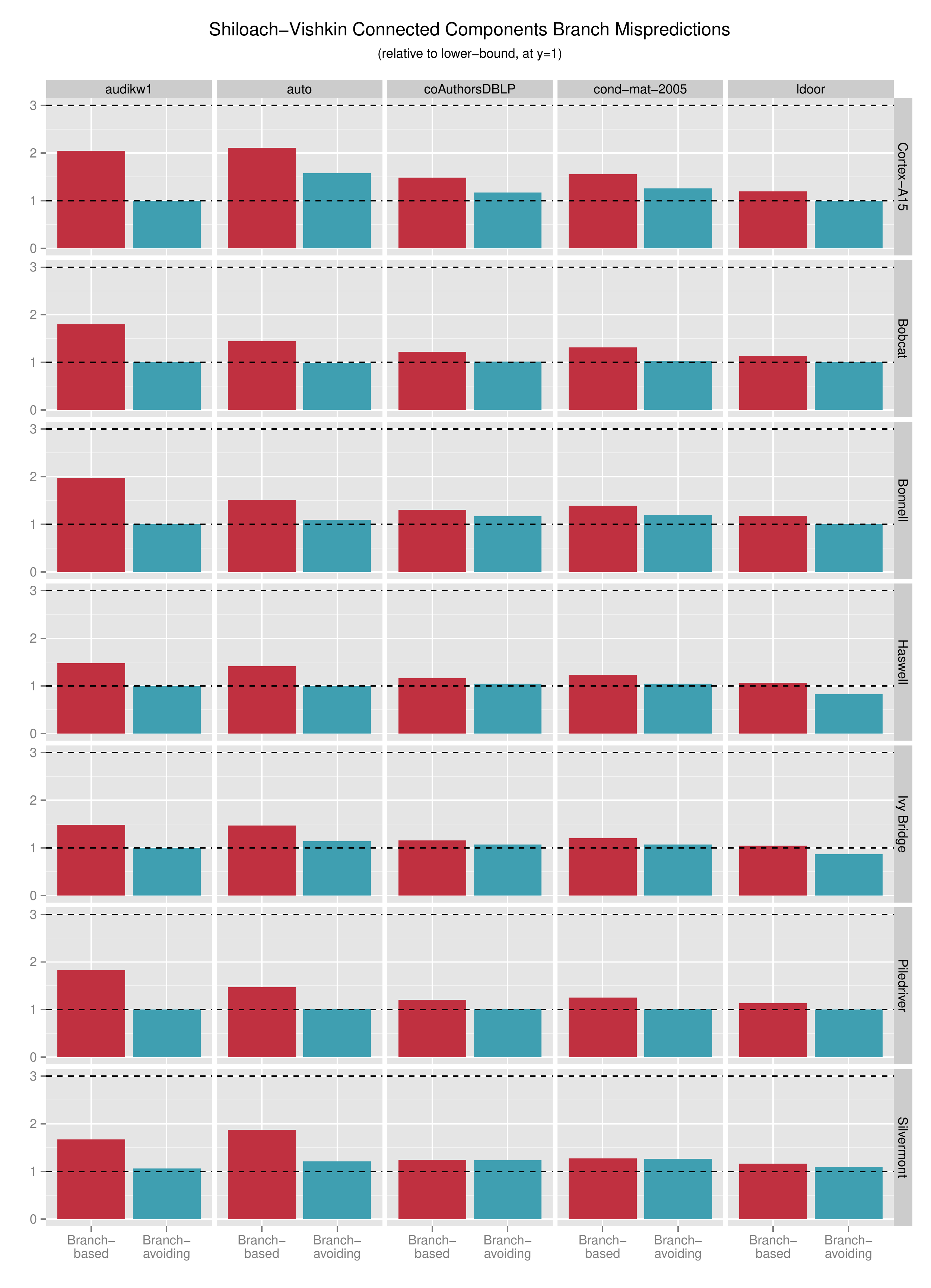}}
		\subfloat[]{\includegraphics[width=0.48\textwidth]{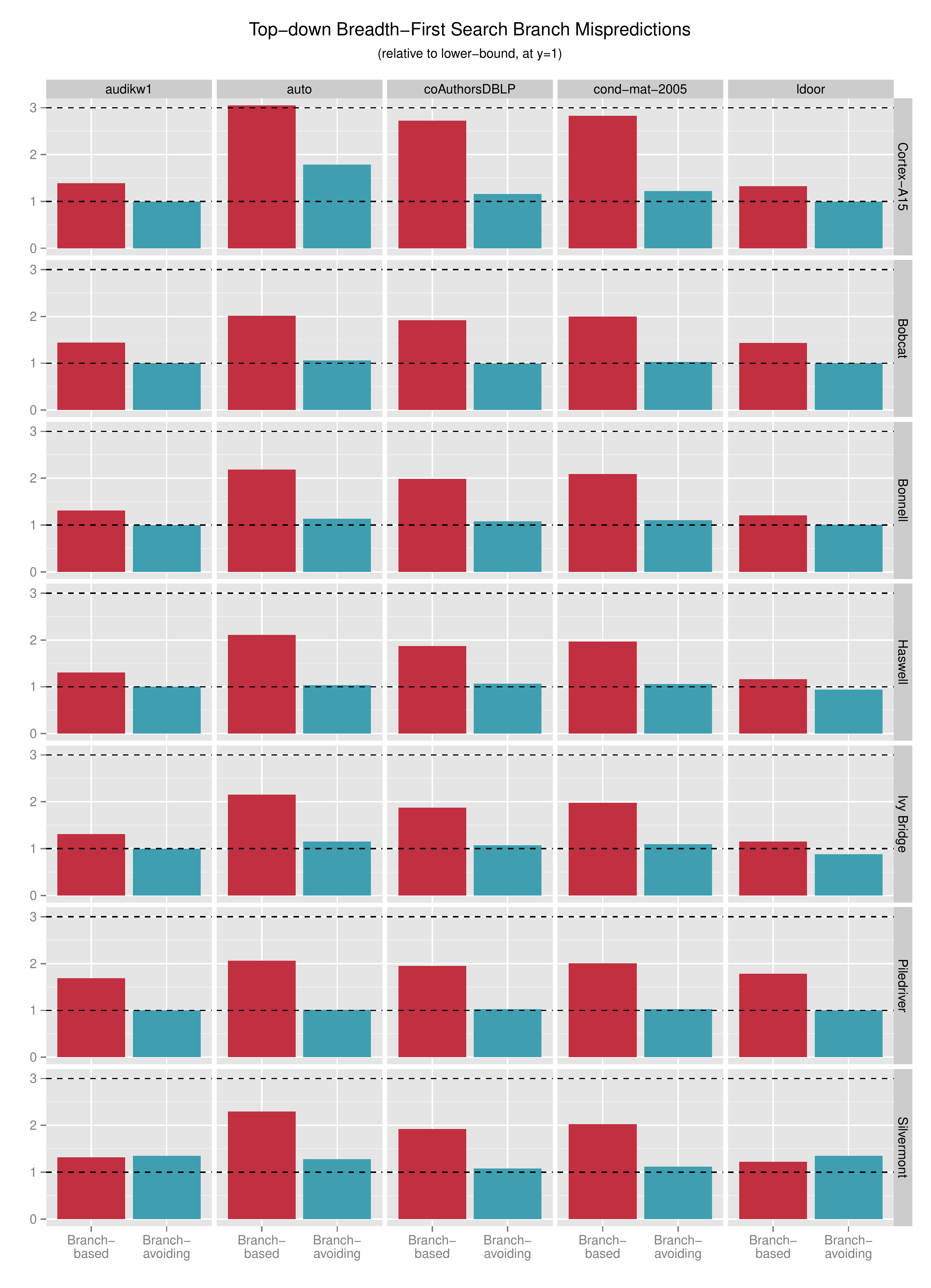}}

	\caption{Lowerbounds on the number of branch mispredictions for SV (based on the model given in 
\RefSection{sec:algSV}) and BFS (based on the model given in \RefSection{sec:algBFS}) . The bars for each algorithm show 
the ratio of branch mispredictions in comparison to the lower bound. For BFS, an upper bound is also provided.}
	\label{fig:mispred--bounds--algs}
\end{figure*}

\begin{figure*}[!t]
	\centering
                \subfloat[]{\includegraphics[width=0.57\textwidth]{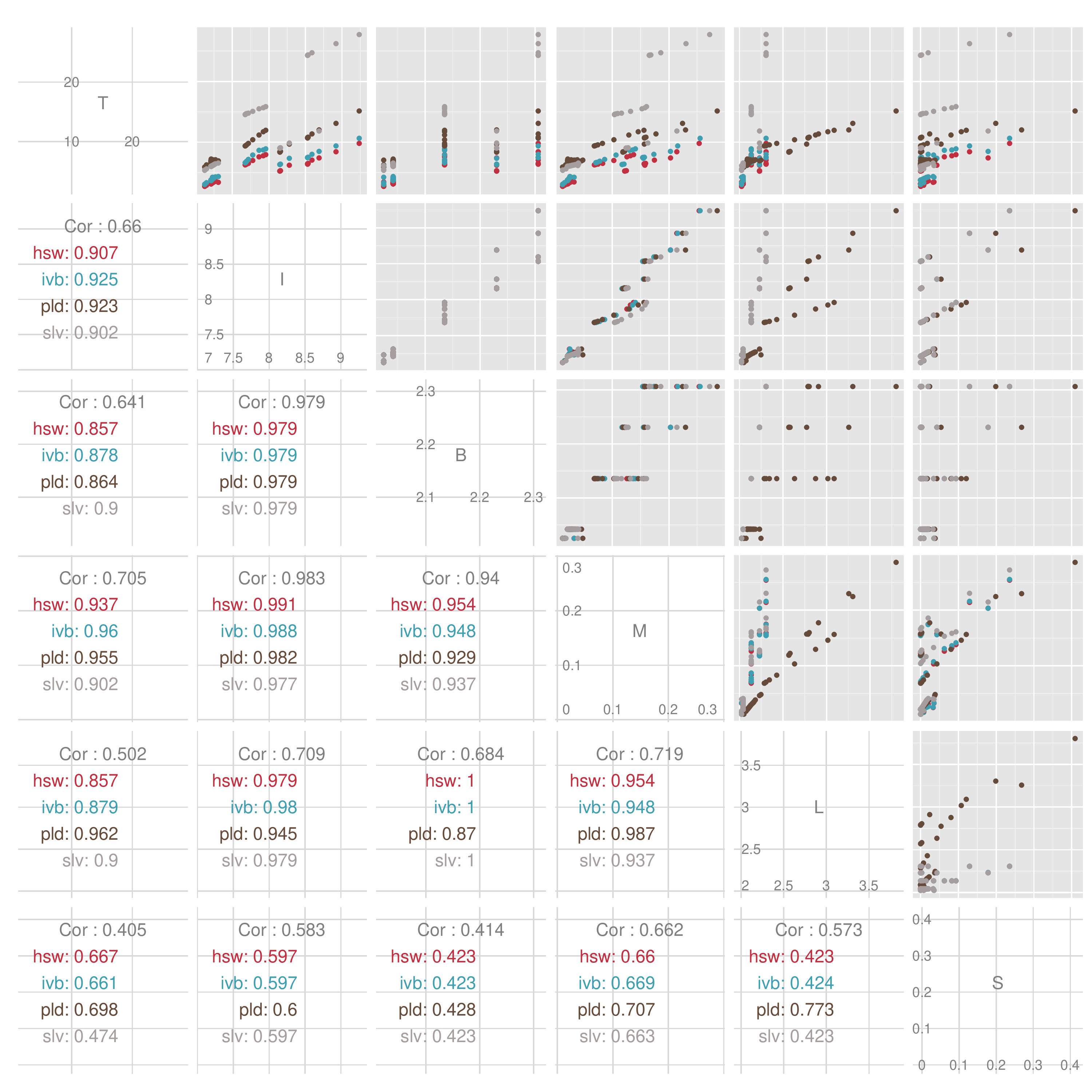}}

                \subfloat[]{\includegraphics[width=0.57\textwidth]{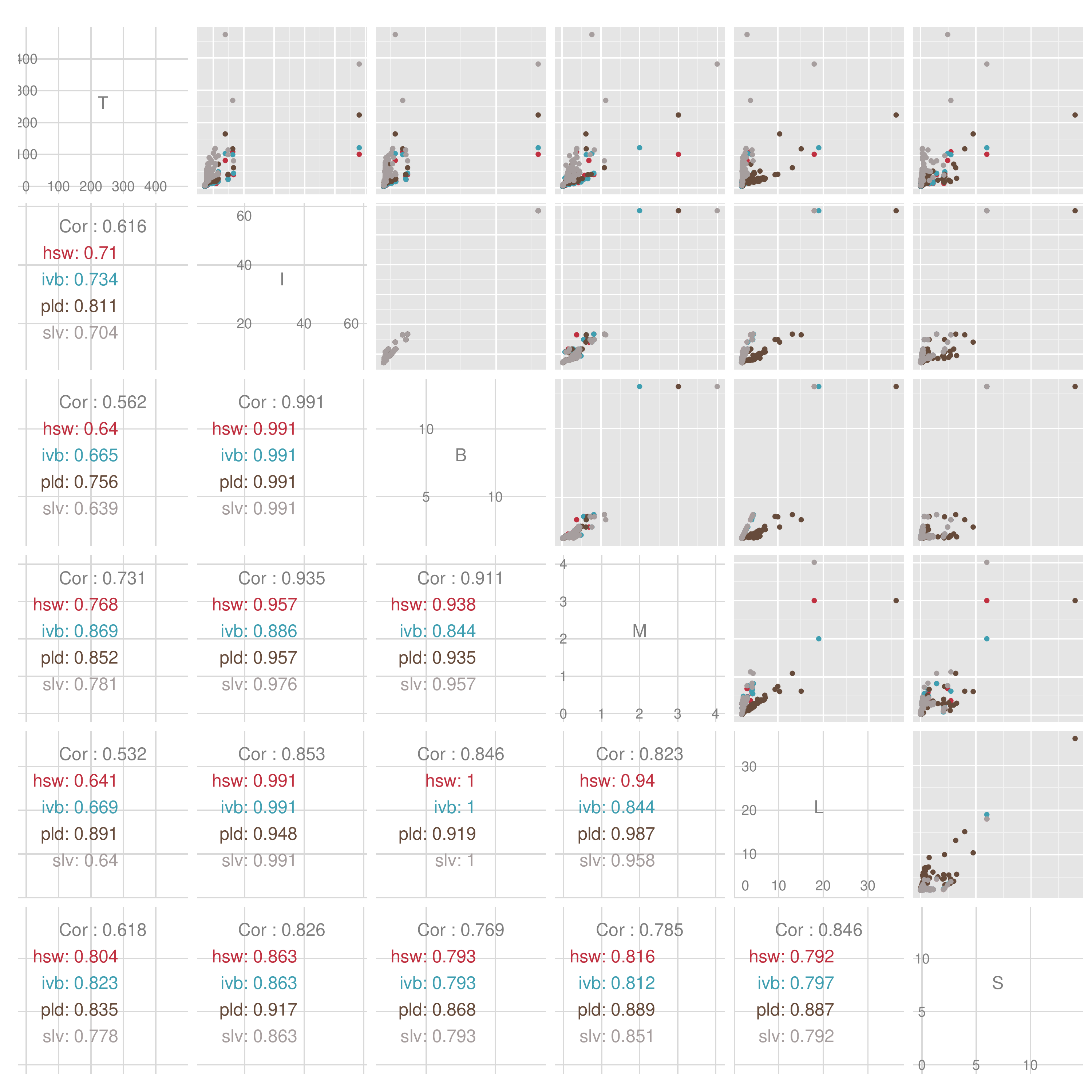}}
	\caption{Correlations among time (T), instructions (I), branches (B), mispredictions (M), load operations (L), 
and store operations (S) per edge, both for (a) SV, and (b) BFS. Each sample is a different iteration/level for one of 
the graph.
        }
        \label{fig:corr}
\end{figure*}

\subsection{Experimental Setup}

For experimental evaluation we tried several implementations of Shiloach-Vishkin connected components algorithm and the top-down breadth-first search algorithm for \textsc{x86-64} and \textsc{ARM} architectures. Unfortunately, compilers do not provide explicit control over the use of branches or conditional moves in the generated code, which complicated our analysis.
The \textsc{x86-64} systems compilers tended to generate conditional branches when conditional moves could be used; albeit we found two ways to make compilers avoid branches, both of them involved inefficiencies: 1) inline assembly allowed manual selection of instructions; however the compilers generated suboptimal code around the inlined assembly; 2) force the compiler to use \textbf{SETcc} instructions by storing the result of the comparison into a byte variable followed by extending the bit mask for conditional selection - this approach caused the compiler to generate $8-9$ instructions where a single \textbf{CMOVcc} instruction could suffice. On the \textsc{ARM} system we had a reverse problem: instead of a conditional branch or move, compilers preferred to use conditional store instructions, which impose big performance penalty on \textit{Cortex-A15}.

For better control of generated code we implemented both connected components and breadth-first search algorithms in \textit{x86-64} and \textsc{ARM} assembly using the Pea\-chPy \cite{dukhanpeachpy} framework. We performed our experiments on $7$ systems with different microarchitectures, these are presented in Table \ref{tab:systems}. On all systems the assembly implementations performed at least as well as C implementations. The algorithms were tested on graphs taken from the DIMACS 10 Graph Challenge \cite{DIMACS10}, detailed in Table \ref{tab:dimacs}.

\subsection{Connected Components}

\RefFigure{fig:sv--time--vs--iteration} depicts the ratio of the execution times as a function of the iteration for the Shiloach-Vishkin algorithm on the systems in \reftab{tab:systems}.  The ratio for each of the algorithm is between the execution time of a given iteration and the fastest iteration of the branch-based algorithm. The branch-based algorithm is depicted by the red curve and the branch-avoiding algorithm is depicted by the blue curve. The abscissa for these figures is the iteration of the algorithm. In each subfigure the total speedup of the branch-avoiding algorithm over the branch-based algorithm is given. Note that for several of the iterations the difference between the branch-based algorithm and the branch-avoiding is high as $30\%-50\%$, with the branch-avoiding algorithm being the faster of these. In a handful of cases, specifically on the \refbonnell system, the branch-based algorithm is $20\%$ faster than the branch-avoiding algorithm.

Recall that as the connected component id propagates in the graph fewer vertices change their connected component - this also makes the branch predictor job easier as it will accurately predict the condition. \RefFigure{fig:sv--branches--vs--iteration} and \RefFigure{fig:sv--mispredictions--vs--iteration} depict the ratio of branches and number of branch mispredictions as a function of the iteration, respectively.

On some systems, for example \refarndale, the branch-avoiding algorithm offers better performance for all iterations of the algorithm over all the graphs. While for some systems, in the initial iterations the branch-avoiding algorithms offers better performance and in the later iterations the branch-based algorithm gives better performance. This is both system- and graph-dependent. In the case that there is a crossover point for performance dominance of the algorithms, it is a single crossover point from where the branch-avoiding algorithm is initially faster to where the branch-based is faster in the later iterations. 
The significance of the single crossover point is that this may allow creating a hybrid algorithm that uses the faster of the two algorithms based on the iteration.

\RefFigure{fig:sv--branches--vs--iteration} shows that the branch-based algorithm has nearly double the number of branches than the branch-avoiding algorithm. For the Intel and AMD systems, the number of branches is constant throughout the iterations while for the \refarndale system it is not. For the Intel and AMD systems, the hardware counter returns the number of retired branch instructions while the \textsc{ARM} system returns the number of dispatched branches. Due to the higher misprediction rate in the first iterations, the number of dispatched branches is also higher as these are flushed instructions.

The branch-based algorithm can potentially have as many as 4X the number of branch mispredictions as that of the branch-avoiding algorithm, \reffig{fig:sv--mispredictions--vs--iteration}. In all cases the branch-avoiding algorithm has fewer branches and branch mispredictions. Note, for most graphs, the ratio between the total number of mispredictions for the two algorithms (denoted by the number at the top-right corner of each subfigure) for a given graph is within a small region for all systems.

\RefFigure{fig:mispred--bounds--algs} (a) depicts the ratio of the total number of branch mispredictions for the two algorithm versus the lower-bound on the number of branch mispredictions. The lower-bound is denoted with a black line at y=1 and this is equal to the lower-bound presented in \refsec{sec:algSV} for the 2-bit branch-predictor.
For most systems, the branch-avoiding algorithms is near the lower-bound, while the branch-based algorithm is well above this line. For the \refarndale system, there are three different graphs in which the branch misprediction rate is well above the lower-bound, for the $auto$ graph the branch misprediction rate is $50\%$ above the lower-bound. This means that implemented branch-predictor in fact increases the misprediction rate. Both the \refbonnell and the \refslv systems also have higher than lower-bound miss rate for several of the graphs. However, these are lower than the miss rate of the \refarndale system.

\subsection{Breadth First Search}

\RefFigure{fig:bfs--time--vs--iteration} depicts the ratio of the execution times as a function of the iteration for BFS.  The ratio for each of the algorithm is between the execution time of a given iteration and the fastest iteration of the branch-based algorithm. The branch-based algorithm is depicted by the red curve and the branch-avoiding algorithm is depicted by the blue curve. The abscissa for these figures is the iteration of the algorithm. In each subfigure the total speedup of the branch-avoiding algorithm over the branch-based algorithm is given. In most cases the branch-avoiding algorithm does not offer a speedup and in fact causes a slowdown for BFS. The branch-avoiding algorithm performs the best on the \refslv system where it is faster for 4 out the 5 test cases.

\RefFigure{fig:bfs--branches--vs--iteration} depicts the ratio of the number of branches of the branch-based algorithm with the branch-avoiding algorithm. 
\RefFigure{fig:bfs--mispredictions--vs--iteration} depicts the ratio of the number of branches mispredictions of the branch-based algorithm with the branch-avoiding algorithm. 
Note the similarity of these figures with the time per iteration, \reffig{fig:bfs--time--vs--iteration}. Also, note that the branch-based algorithm has nearly double the number of branches than the branch-avoiding algorithm.
If one considers the average number of branches per edge traversal in each frontier, these ratios are still maintained.

\RefFigure{fig:mispred--bounds--algs} (b) depicts the ratio of the total number of branch mispredictions for the two algorithm versus the lower-bound on the number of branch mispredictions. The lower-bound is denoted with a black line at y=1 and the upper-bound is denoted with a black line at y=3. The upper and lower bounds on the number of branch mispredcition for a 2-bit branch predictor BFS was discussed in \refsec{sec:algSV}. The lower-bound is dependent on the number of vertices found in the traversal, $\hat{V}$ and the upper bound is 3 times the lower-bound. Similar to the Shiloach-Vishkin algorithm, the branch-avoiding algorithms is near the lower-bound for most graphs and on most systems. Again, the \refarndale system has some of the higher misprediction rates for both algorithms.

\subsection{The effects of misprediction}

To get an idea of how strongly mispredictions influence performance, we show pairwise correlations among the \emph{a priori} most likely predictors of execution time: instructions, loads, stores, and based on the subject of this paper, branches and branch mispredictions.
\RefFigure{fig:corr} shows this data for the branch-based versions of both SV (left half) and BFS (right half).
Each $6\times 6$ grid of subplots shows the correlations among time, instructions, branches, mispredictions, loads, and stores, measured per edge traversal.
For example, (row 1, column 2) subplot in each half is a scatter plot comparing time (``T'') on the y-axis with instructions (``I'') on the x-axis.
The points are color-coded by platform, on the subset of platforms that supported all necessary hardware performance counters.
For each $(R, C)$ plot in the upper-triangle, the computed correlation coefficients appear in the transposed $(C, R)$ position of the lower-triangle.

In the case of SV, mispredictions more strongly correlate with time than instructions, branches, loads, and stores.
Though not a strict proof-of-cause, this observation is nevertheless somewhat surprising, as it implies mispredictions may be nearly or even more important than memory behavior.
By contrast, in the case of BFS, the correlations with stores and mispredictions is roughly equal, with stores being slightly more strongly correlated than time.
This confirms the performance behavior seen previously, namely, that eliminating branches at the cost of increasing stores should not be expected to improve performance. In practice the number of stores were increased by two order of magnitude for some graphs.

\section{Conclusions}
\label{sec:conclusions}

On the one hand, our study is a positive result for the branch-avoiding technique in the case of SV, where mispredictions are more strongly correlated to time than even memory traffic, much to our surprise.
This raises the question of whether branch-avoidance might be important in other computations, and whether increased microarchitectural support for predication-like instructions might have more significant benefits.

On the other hand, our study is a negative result for BFS.
Stores are as critical as branch mispredictions, so the tradeoff that reduces branches at the cost of significantly increasing stores cannot pay off.
One question is why: although total stores increased by much as $100\times$, the actual slowdown was always $2\times$ or less.
Indeed, the extra stores are purely ``local'' in that they should mostly hit in cache, by design of the implementation.
Thus, there is a potential in the microarchitecture to address whatever resource constraints the additional stores impose, such as buffers for more outstanding operations.

An additional question is to what extent compiler transformations can or should replicate the transformations we implemented by hand.
Though not shown explicitly, although the \textsc{ARM} system supported predicated instructions, no compiler produced transformations equivalent to our hand-generated code.
Whether the gap can be filled remains open, in our view.




%

\bibliographystyle{abbrv}

\bibliography{gea}  


\end{document}